\documentclass[11pt]{article}
\usepackage{fullpage}

\usepackage{times}
\usepackage{comment,amsfonts,amssymb,amsmath,amsthm,graphicx,algorithm,algorithmic}
\newcommand{\commentout}[1]{}

\ifx\pdftexversion\undefined
\usepackage[colorlinks,linkcolor=black,filecolor=black,citecolor=black,urlco
lor=black,pdfstartview=FitH]{hyperref}
\else
\usepackage[colorlinks,linkcolor=blue,filecolor=blue,citecolor=blue,urlcolor
=blue,pdfstartview=FitH]{hyperref}
\fi

\newcommand{\alert}[1]{\textbf{\color{red}
[[[#1]]]}\marginpar{\textbf{\color{red}**}}\typeout{ALERT:
\the\inputlineno: #1}}

\newcommand{\ex}{{\cal EXP}}

\newcommand{\etal}{\emph{et. al. }}

\newcommand{\mommit}[1]{}
\newcommand{\namedref}[2]{\hyperref[#2]{#1~\ref*{#2}}}
\newcommand{\sectionref}[1]{\namedref{Section}{#1}}

\newcommand{\theoremref}[1]{\namedref{Theorem}{#1}}

\newcommand{\claimref}[1]{\namedref{Claim}{#1}}
\newcommand{\lemmaref}[1]{\namedref{Lemma}{#1}}

\newcommand{\obref}[1]{\namedref{Observation}{#1}}

\newtheorem{theorem}{Theorem}
\newtheorem{lemma}{Lemma}
\newtheorem{corollary}[lemma]{Corollary}

\newtheorem{claim}[lemma]{Claim}

\newtheorem{observation}[lemma]{Observation}

\def\cE{{\cal E}}
\def\cP{{\cal P}}
\def\cG{{\cal G}}
\def\Diam{{\mathit{Diam}}}
\def\WDiam{{\mathit{WeakDiam}}}

\usepackage{pdfsync}
\usepackage{authblk}
\begin{document}

\title{Distributed Strong Diameter Network Decomposition}

\author[1]{Michael Elkin}
\author[1]{Ofer Neiman}

\affil[1]{Department of Computer Science, Ben-Gurion University of the Negev,
Beer-Sheva, Israel. Email: \texttt{\{elkinm,neimano\}@cs.bgu.ac.il}}

\date{}
\maketitle

\begin{abstract}
For a pair of positive parameters $D,\chi$, a partition $\cP$ of the vertex set $V$ of an $n$-vertex graph $G = (V,E)$ into disjoint clusters of diameter at most $D$ each is called a {\em $(D,\chi)$ network decomposition}, if the supergraph $\cG(\cP)$, obtained by contracting each of the clusters of $\cP$, can be properly $\chi$-colored.
The decomposition $\cP$ is said to be {\em strong} (resp., {\em weak}) if each of the clusters has strong (resp., weak) diameter at most $D$, i.e., if for every cluster $C \in \cP$ and every two vertices $u,v \in C$, the distance between them in the induced graph $G(C)$ of $C$ (resp., in $G$) is at most $D$.

Network decomposition is a powerful construct, very useful in distributed computing and beyond. It was introduced by Awerbuch \etal \cite{AGLP89} in the end of the eighties.
These authors showed that strong $(2^{O(\sqrt{\log n\log\log n})},2^{O(\sqrt{\log n\log\log n})})$ network decompositions can be computed in
$2^{O(\sqrt{\log n\log\log n})}$ distributed time.
Their result was improved at the beginning of nineties by Panconesi and Srinivasan \cite{PS92}, who showed that $2^{O(\sqrt{\log n\log\log n})}$ in all the three expressions can be replaced  by $2^{O(\sqrt{\log n})}$. Around the same time Linial and Saks \cite{LS93}  devised an ingenious randomized algorithm that constructs {\em weak} $(O(\log n),O(\log n))$ network decompositions in $O(\log^2 n)$ time. It was however open till now if {\em strong} network decompositions with both parameters $2^{o(\sqrt{\log n})}$ can be constructed in distributed $2^{o(\sqrt{\log n})}$ time.

In this paper we answer this long-standing open question in the affirmative, and show that strong $(O(\log n),O(\log n))$ network decompositions can be computed in $O(\log^2 n)$ time. We also present a tradeoff between parameters of our network decomposition.
Our work is inspired by and relies on the ``shifted shortest path approach", due to Blelloch \etal \cite{BGKMPT11}, and Miller \etal \cite{MPX13}. These authors developed this approach for PRAM algorithms for padded partitions. We adapt their approach to network decompositions in the distributed model of computation.

\end{abstract}

\thispagestyle{empty}
\newpage
\setcounter{page}{1}

\section{Introduction}
\subsection{Definitions and Motivation}

Consider an unweighted undirected $n$-vertex graph $G = (V,E)$, and suppose that it models a communication network. Each vertex hosts a processor with a distinct identity number from the range $\{1,\ldots,n\}$, and these processors communicate with one another via the edges of $G$ in synchronous rounds. The running time of an algorithm in this model is the number of rounds of distributed communication.

In the coloring problem one wishes to compute a proper coloring $\varphi$ of $G$ that employs a small number of colors. A coloring $\varphi$ is said to be {\em proper} if for every edge $(u,v) \in E$, we have $\varphi(u) \neq \varphi(v)$. In a seminal paper \cite{AGLP89}, Awerbuch \etal introduced a generalization of vertex coloring, in which one can cluster vertices of $G$ into clusters of small diameter. A partition $\cP$ of $G$ into disjoint clusters induces a supergraph $\cG(\cP) = (\cP,\cE)$, where
$$\cE = \{(C,C') \mid C,C' \in \cP, C \neq C', \exists (v,v') \in E \cap (C \times C')\}~.$$
A partition $\cP$ is called a {\em strong} (respectively, {\em weak}) {\em network decomposition} of $G$ with parameters $D$ and $\chi$, or shortly, {\em $(D,\chi)$ network decomposition},  if all clusters of $\cP$ have strong (resp., weak) diameter at most $D$, and the supergraph $\cG(\cP)$ can be properly colored with at most $\chi$ colors. Note that an ordinary proper $\chi$-coloring can be viewed as a $(0,\chi)$ network decomposition.

The {\em strong} (respectively, {\em weak}) {\em diameter} of a cluster $C$ is defined by $\Diam(C) = \max_{v,v' \in C} d_{G(C)}(v,v')$ (resp., $\WDiam(C) = \max_{v,v' \in C} d_G(v,v')$).
The notation $d_G$ (respectively, $d_{G(C)}$) denotes the distance function in $G$ (resp., in the induced subgraph $G(C)$ of $C$).
The strong (resp., weak) diameter of a partition $\cP$ is the maximum strong (resp., weak) diameter of its clusters.

Network decomposition is a very powerful construct in distributed computing. The original motivation of \cite{AGLP89} was symmetry breaking problems, such as maximal independent set, maximal matching and $(\Delta+1)$-vertex-coloring, where $\Delta$ is the maximum degree of the input graph.
Given a $(D,\chi)$ network decomposition $\cP$ along with a $\chi$-coloring of the induced supergraph $\cG(\cP)$, each of these problems can be solved within $O(D \cdot \chi)$ time. This is done by solving them in parallel on each of the clusters of color class 1, then extending the solution to each of the clusters of color class 2, etc.
Since clusters within each color class are at least 2 apart one from another, computations within the same color class can be conducted in parallel. Moreover, since the maximum clusters' diameter is bounded by $D$, one can perform each of these $\chi$ phases within $O(D)$ time by a naive algorithm.  (The naive algorithm collects the entire cluster's topology into a central vertex, solves the problem locally, and disseminates  the solution to all vertices of the given cluster.)

Later additional applications of network decompositions were discovered. Dubhashi \etal \cite{DMPRS05} used network decompositions for computing sparse spanners and linear-size skeletons. Barenboim \etal \cite{Bar12,BEG15} devised distributed approximation algorithm for the graph coloring and minimum dominating set problems, which employ network decompositions. Network decompositions are also closely related to {\em neighborhood covers}, which are used extensively for routing \cite{AP92} and synchronization \cite{Awe85,APPS92}.  The relationship between neighborhood covers and network decompositions was explored in \cite{ABCP92}.
Barenboim \etal \cite{BEG15} have also showed that network decompositions can be used to build low-intersecting partitions, which are, in turn, used for computing universal Steiner trees \cite{BDRRS12}.

To summarize, network decompositions have numerous applications in distributed computing and beyond. They also constitute a very appealing combinatorial  construct, well worth studying on its own right.

\subsection{Previous and Our Results}

Awerbuch \etal \cite{AGLP89} devised a deterministic algorithm with running time $2^{O(\sqrt{\log n\log\log n})}$, that computes a strong
$(2^{O(\sqrt{\log n\log\log n})}, 2^{O(\sqrt{\log n\log\log n})})$ network decomposition. This result was improved by Panconesi and Srinivasan \cite{PS92}, whose algorithm has running time $2^{O(\sqrt{\log n})}$, and both parameters of the decomposition of \cite{PS92} are $2^{O(\sqrt{\log n})}$ as well.
In another seminal work, titled ``Low Diameter Graph Decompositions",  Linial and Saks \cite{LS93} conducted a systematic investigation of network decompositions.
They showed that for any $k\le \log n$, every $n$-vertex graph admits a strong $(2k-2, 2n^{1/k}\log n)$
network decomposition, and for any $\lambda \le \log n$, it admits a strong $(2 n^{1/\lambda} \log n,\lambda)$ network decomposition, and that these bounds are nearly tight. They have also devised a randomized distributed algorithm for computing {\em weak} network decompositions in expected time $O(k\cdot n^{1/k} \cdot \log n)$, with essentially the same parameters.
In particular, and most notably, for $k = \log n$, their algorithm produces a {\em weak} $(O(\log n), O(\log n))$ network decomposition in $O(\log^2 n)$ time.

Remarkably, quarter a century after the SODA'91 publication  of Linial and Saks' paper, their algorithm is still the only algorithm whose running time is at most polylogarithmic in $n$, and which produces a network decomposition with both  parameters being at most polylogarithmic in $n$. Moreover, so far it was not known if such a result can be achieved for {\em strong} network decompositions. Linial and Saks \cite{LS93} themselves posed this as an open problem. Specifically, near the end of the introduction of \cite{LS93} they wrote:
\\
\\
{\it ``We note that we do not know how to make a similar guarantee on the strong diameter."}
\\
\\
In this paper we resolve this long-standing  open question in the affirmative. We devise a randomized algorithm with running time $O(\log^2 n)$ that computes a {\em strong} $(O(\log n), O(\log n))$ network decomposition. Moreover, similarly to Linial and Saks \cite{LS93}, we can also trade between the parameters. Specifically, for any $k\le \log n$, our randomized algorithm has running time $O(n^{1/k} \cdot k^2)$ and computes a {\em strong} $(2k-2, O(k\cdot n^{1/k}))$ network decomposition. In the other regime, for any $\lambda\le\log n$, in time $O(\lambda\cdot n^{1/\lambda} \cdot \log n)$ we compute a strong $(O( n^{1/\lambda} \log n),\lambda)$ network decomposition. Note that the number of colors and running time are slightly better than those of \cite{LS93} in the first regime. As in \cite{LS93}, all messages sent in our algorithm consist of $O(1)$ words.

The main technique that made our result possible is the ``shifted shortest path approach", due to Blelloch \etal \cite{BGKMPT11}, and Miller \etal \cite{MPX13}.
These authors  developed  this approach for  computing padded partitions in the PRAM model. Specifically, Miller \etal \cite{MPX13} devised a PRAM algorithm for computing a {\em strong padded partition}, i.e., a partition with strong diameter at most $O(\log n)/\beta$, for a parameter $\beta\le 1/2$, and such that the fraction of edges that cross between different clusters of the partition is at most $\beta$.

It is known that padded partitions are related to network decompositions.  This relationship was exploited by Bartal \cite{Bar96}, who showed that the approach of Linial and Saks \cite{LS93} for constructing network decompositions can be used to build padded partitions. In this work we exploit this relationship in the opposite direction, and show that Miller's \etal \cite{MPX13} approach for constructing padded partitions can be used for building network decompositions. Our algorithm is similar in spirit to the algorithm of \cite{LS93}, in which every vertex $v$ samples a radius $r_v$ from a geometric (or exponential, in our case) distribution, and broadcasts this to its $r_v$-neighborhood. The main difference is in determining the clusters: While in \cite{LS93} a vertex $x$ decides to join a cluster centered at $v$ if $v$ has the minimal ID among broadcasts that reached $x$, and furthermore $r_v$ is strictly larger than the distance $d(x,v)$ (this is the distance in the current graph). In our algorithm, we do not use IDs, we let $x$ compare the shifted random variables $r_v-d(x,v)$ for all vertices $v$ whose broadcast reached it, and decide according to the difference between the largest and the second largest values. This idea is inspired by \cite{MPX13}, who use a similar comparison in the {\em analysis} of their algorithm for padded partitions. However, the fact that this algorithm yields a strong diameter is somewhat more involved in our setting.

\subsection{Related Work}

Barenboim \etal \cite{BEG15} devised a randomized constant time algorithm for constructing  strong $(O(1),n^\epsilon)$ network decompositions, for an arbitrarily small constant $\epsilon>0$. Kutten \etal \cite{KNPR14} extended the algorithm of Linial and Saks \cite{LS93} for constructing  network decompositions to hypergraphs.  A long line of research developed  network decompositions  for graphs of bounded growth, see, e.g.,  \cite{GV07,KMW05,SW08}.

\section{Distributed Algorithm for Strong Diameter Network Decomposition}\label{sec:main}

Here we prove our main result. For a more accessible presentation, we first show a simpler version, and improve the number of colors in \sectionref{sec:less-color}.

\begin{theorem}\label{thm:main}
For any unweighted graph $G=(V,E)$ on $n$ vertices, and parameters $1\le k\le \ln n$, $3<c$, our randomized distributed algorithm computes, with probability at least $1-3/c$, a {\em strong} $(2k-2,(cn)^{1/k}\cdot\ln (cn))$ network decomposition of $G$. The number of rounds required is $k(cn)^{1/k}\cdot\ln (cn)$, and each message consists of $O(1)$ words.
\end{theorem}

Note that taking $c=2^k$ does not affect the number of blocks and rounds by more than a constant factor.
Following \cite{LS93}, we form the partition by carving blocks. A {\em block} $W\subseteq V$ is set of vertices, and the connected components of $G(W)$ are clusters. Clearly, these clusters form an independent set in ${\cal G}({\cal P})$, and thus can be colored with a single color. So the chromatic number of ${\cal G}({\cal P})$ is bounded by the number of blocks our algorithm generates.

\paragraph{Construction.}
The algorithm is a subtle modification of the \cite{LS93} algorithm, inspired by the recent methods of \cite{MPX13}. Let $\beta=\ln (cn)/k$. The algorithm consists of phases $t=1,2,\ldots,\lambda$, for $\lambda=(cn)^{1/k}\cdot\ln (cn)$. Let $G_1=G$. In each phase $t$ we carve a block $W_t$ out of the current graph $G_t$, and let $G_{t+1}=G_t\setminus W_t$.

To implement the $t$-th phase, every vertex $v\in V(G_t)$ chooses independently in parallel a value $r_v^{(t)}$ (we shall omit the superscript whenever it is clear from context), by sampling from the exponential distribution with parameter $\beta$, denoted $\ex(\beta)$, which has density
\[
f(x)=\left\{\begin{array}{ccc} \beta\cdot e^{-\beta x} & x\ge 0\\
0 & \text{otherwise.} \end{array} \right.
\]

For $v\in V$, let ${\cal E}_v$ be the event that at some phase $t$, $r_v^{(t)}\ge k+1$. We will later prove the following lemma.
\begin{lemma}\label{lem:radius}
With probability at least $1-2/c$, none of the events ${\cal E}_v$ hold.
\end{lemma}

Every vertex $v$ will broadcast the value $r_v$ to every vertex of $G_t$ within distance $R_v:=\lfloor r_v\rfloor$ from it. Note that assuming \lemmaref{lem:radius}, $R_v\le k$. 
Each vertex $y$ in $G_t$ records the values of $r_v$ for vertices $v$ whose broadcast reached $y$, and also the distances in $G_t$ to these vertices. Then $y$ orders these vertices $v_1,\dots,v_s$ in non-increasing order according to $m_i=r_{v_i}-d_{G_t}(y,v_i)$.
We declare that $y$ joins $W_t$ iff $m_1-m_2>1$.
Observe that all $m_i$ are nonnegative, since $y$ will hear the broadcast of $v_i$ only if $d_{G_t}(y,v_i)\le R_{v_i}$, the latter is at most $r_{v_i}$.
If $s=1$, i.e. there is no second broadcast that reached $y$, define $m_2=0$ (observe $m_1$ is well defined as $y$ also broadcasts). If indeed $y$ joins $W_t$, then we say that $y$ {\em chose} the center $v_1$.

We begin by analyzing the strong diameter of the blocks.
\begin{observation}\label{ob:ser}
If $y$ chose $v_1$ as a center at phase $t$, then $d_{G_t}(v_1,y) <r_{v_1}-1$.
\end{observation}
\begin{proof}
If $d_{G_t}(v_1,y)\ge r_{v_1}-1$, then $m_1\le 1$, which implies that $m_1-m_2\le 1$, contradicting the fact that $y$ joins $W_t$.
\end{proof}

\begin{claim}\label{claim:strong}
If a vertex $y\in V(G_t)$ chose $v$ at phase $t$, then every vertex $x$ on the shortest-path from $v$ to $y$ in $G_t$ must have chosen $v$ at phase $t$ as well.
\end{claim}
\begin{proof}
Since $d_{G_t}(v,x)\le d_{G_t}(v,y)$, the broadcast of $v$ at phase $t$ must have reached $x$ as well, so $x$ records the value $m=r_v-d_{G_t}(x,v)$. Seeking contradiction, assume $x$ did not choose $v$, then there exists $v'$ for which $x$ records the value $m'=r_{v'}-d_{G_t}(x,v')$ with $m'\ge m-1$ (if there is no such $v'$, then $x$ would have joined $W_t$ with $v$ as center). In particular,
\begin{equation}\label{eq:1}
d_{G_t}(x,v')\le r_{v'}-r_v+d_{G_t}(x,v)+1~.
\end{equation}
It follows that
\begin{eqnarray}\label{eq:arra}
d_{G_t}(y,v')&\le& d_{G_t}(y,x)+d_{G_t}(x,v')\nonumber\\
&\stackrel{\eqref{eq:1}}{\le}&d_{G_t}(y,x)+r_{v'}-r_v+d_{G_t}(x,v)+1\nonumber\\
&=&(d_{G_t}(y,v)-r_v+1)+r_{v'}\\
&<& r_{v'}
\end{eqnarray}
where the last inequality uses \obref{ob:ser}. Thus $d_{G_t}(y,v')\le R_{v'}$, so the broadcast of $v'$ will reach $y$, and $y$ will record a corresponding value of
\[
r_{v'}-d_{G_t}(y,v')\stackrel{\eqref{eq:arra}}{\ge} r_v-d_{G_t}(y,v)-1~,
\]
that is, it is within 1 of the value $y$ stored for $v$,
which contradicts the fact that $y$ chose $v$.
\end{proof}

\begin{lemma}\label{lem:diam}
For every $1\le t\le \lambda$, the block $W_t$ has strong diameter at most $2k-2$.
\end{lemma}
\begin{proof}
Fix any cluster $C$ which is a connected component of $G(W_t)$.
We first argue that if all vertices in $C$ chose the same center $v$, then its strong diameter is at most $2k-2$. To see this, note that by \obref{ob:ser} all vertices $y\in C$ are within $r_v-1$ distance from $v$, since the graph is unweighted, this is at most $R_v-1\le k-1$ (assuming the event of \lemmaref{lem:radius} holds). By \claimref{claim:strong}, every vertex on the shortest-path from $v$ to $y$ (in $G_t$) is also included in $C$, so the strong diameter is at most $2k-2$.

Consider now the case that there are two vertices $y,z\in C$ that chose different centers $v,u$. We will show that this assumption must lead to a contradiction. Note we may assume that $y,z$ are adjacent, since for any two non-adjacent $y',z'$ who chose different centers, we can simply walk on the path in $C$ (which is connected) from $y'$ to $z'$ until we find adjacent vertices with a center change occurring.
W.l.o.g assume $y$ is the vertex which recorded the larger value, that is,
\begin{equation}\label{eq:ffd}
r_v-d_{G_t}(y,v)\ge r_u-d_{G_t}(z,u)~.
\end{equation}
By the triangle inequality and \obref{ob:ser} we see that $d_{G_t}(z,v)\le d_{G_t}(y,v)+1< r_v$, which implies $d_{G_t}(z,v)\le R_v$, so that the broadcast of $v$ will reach $z$. The value $z$ obtains from $v$ is
\[
r_v-d_{G_t}(z,v)\ge r_v-(d_{G_t}(y,v)+1)\stackrel{\eqref{eq:ffd}}{\ge} r_u-d_{G_t}(z,u)-1~,
\]
which contradicts the assumption that $z$ chose $u$.
\end{proof}

We next show that $\lambda$ phases suffice to exhaust the graph, which gives this bound on the number of blocks.
To this end, we use the following result from \cite[Lemma 4.4]{MPX13} on the order statistics of shifted exponential random variables.
\begin{lemma}[\cite{MPX13}]\label{lem:prob}
Let $d_1\le\ldots\le d_q$ be arbitrary values and let $\delta_1,\dots,\delta_q$ be independent random variables picked from $\ex(\beta)$. Then the probability that the largest and the second largest values of $\delta_j-d_j$ are within 1 of each other is at most $1-e^{-\beta}$.~\footnote{We state here a special case of their result. The assertion in \cite{MPX13} gives the bound $O(\beta)$, but their proof in fact yields the stronger bound given here.}
\end{lemma}
We use this result to prove the following:
\begin{claim}\label{claim:ref}
For any $y\in V$, and $1\le t'\le \lambda$,
\[
\Pr[y\in G_{t'+1}]\le (1-(cn)^{-1/k})^{t'}~.
\]
\end{claim}
\begin{proof}
Fix any $1\le t\le t'$, and any possible graph $G_t$ such that $y\in V(G_t)$. Let $v_1,\dots,v_q$ be the vertices of $G_t$ that are in the same connected component of $G_t$ with $y$. Let $d_j=d_{G_t}(v_j,y)$, and $\delta_j=r_{v_j}$ (where each $r_{v_j}$ is sampled independently from $\ex(\beta)$). Recall that $y\in W_t$ iff the maximum value among $\delta_j-d_j$ is larger than the second largest by more than 1 (additively). Applying \lemmaref{lem:prob}, we conclude that the probability a vertex $y\in V(G_t)$ joins $W_t$ is at least $e^{-\beta}=(cn)^{-1/k}$ (this holds even in the event that no other broadcast reached $y$, by definition of $\ex(\beta)$). Since this bound holds regardless of the outcome of previous phases,
\[
\Pr[y\in G_{t'+1}]=\Pr\left[\bigcap_{t=1}^{t'}\{y\notin W_t\}\right]=\prod_{t=1}^{t'}\Pr[y\notin W_t\mid y\notin W_1,\dots,y\notin W_{t-1}]\le (1-(cn)^{-1/k})^{t'}
\]
\end{proof}
\begin{corollary}
With probability at least $1-1/c$, $G_{\lambda+1}$ is empty.
\end{corollary}
\begin{proof}
Using \claimref{claim:ref} with $t'=\lambda=(cn)^{1/k}\cdot\ln (cn)$, we see that the probability a vertex $y$ did not join any block is at most $(1-(cn)^{-1/k})^\lambda\le 1/(cn)$. Applying the union bound on the $n$ vertices, we get that with probability $1-1/c$, within $\lambda$ phases the graph is indeed exhausted.
\end{proof}

We are now ready to prove \lemmaref{lem:radius}.

\begin{proof}[Proof of \lemmaref{lem:radius}]
Fix any $v\in V$.
Since each $r_v$ is sampled independently from $\ex(\beta)$, we have for any $1\le t\le \lambda$, $\Pr[r_v^{(t)}\ge k+1]=e^{-\beta(k+1)}$. By using \claimref{claim:ref} with $t'=i\cdot (cn)^{1/k}$ (for some $0\le i\le \ln(cn)$), we obtain $\Pr[v\in G_{t'+1}]\le e^{-i}$. Now,
\begin{eqnarray*}
\Pr[{\cal E}_v]&\le& \sum_{t=1}^\lambda\Pr[r_v^{(t)}\ge k+1\mid v\in G_t]\cdot\Pr[v\in G_t]\\
&\le&\sum_{i=0}^{\ln(cn)}\sum_{t=1}^{(cn)^{1/k}}\Pr[r_v^{(i\cdot(cn)^{1/k}+t)}\ge k+1\mid v\in G_{i\cdot(cn)^{1/k}+t}]\cdot\Pr[v\in G_{i\cdot(cn)^{1/k}+1}]\\
&\le&\sum_{i=0}^{\ln(cn)}e^{-i}\cdot\sum_{t=1}^{(cn)^{1/k}} e^{-\beta(k+1)}\\
&\le&\sum_{i=0}^{\ln(cn)}e^{-i}\cdot(cn)^{1/k}\cdot (cn)^{-1-1/k}\\
&\le&2/(cn)~.
\end{eqnarray*}
The lemma follows from a union bound over the $n$ vertices.
\end{proof}

We conclude by analyzing the running time and messages size. Note that there are $\lambda=(cn)^{1/k}\cdot\ln(cn)$ phases, and each phase requires $k$ rounds (assuming \lemmaref{lem:radius}), so the total number of rounds is as promised. We claim that our algorithm can in fact be implemented efficiently also in the CONGEST model, where messages must be of size at most $O(\log n)$ bits. This follows since at every round, every vertex can sort the values $m_i$ it has so far, and send to its neighbors only the top two from its list. This is because the values $\lfloor m_i\rfloor$ determine the remaining range the message of $v_i$ needs to be forwarded to, and clustering decisions are based only on the largest two values, so the third and onward values in $v$'s list will not be used by any other vertex.

\subsection{Improved Number of Blocks}\label{sec:less-color}

Here we show how to improve the bound on the number of colors to $O(k\cdot n^{1/k})$, and prove the following.
\begin{theorem}\label{thm:less blocks}
For any unweighted graph $G=(V,E)$ on $n$ vertices, and parameters $1\le k\le \ln n$, $5<c$, our randomized distributed algorithm computes, with probability at least $1-5/c$, a strong $(2k-2,4k(cn)^{1/k})$ network decomposition of $G$. The number of rounds required is $O(k^2(cn)^{1/k})$, and each message consists of $O(1)$ words.
\end{theorem}

The main difference from the previous construction of is that the parameter $\beta$ of the exponential distribution will change at certain points. There will be $\ln n$ stages, each stage consists of a certain number of phases in which we use the same value of $\beta$. The first stage lasts $s_0=2(cn)^{1/k}$ phases in which we use $\beta_0=\ln (cn)/k$. The next stage lasts $s_1=2(cn/e)^{1/k}$ phases, in which we use $\beta_1=\ln (cn/e)/k$. In general, the $i$-th stage lasts $s_i=2(cn/e^i)^{1/k}$ phases, and we use $\beta_i=\ln (cn/e^i)/k$ in these phases.
For $0\le i\le \ln n$, denote by $J_i$ the set of phases in the $i$-th stage, that is, $J_i=\{\sum_{j=0}^{i-1}(s_j)+1,\dots,\sum_{j=0}^is_j\}$.

The total number of phases, which bounds the number of colors needed, is thus
\[
\sum_{i=0}^{\ln n}s_i=2\sum_{i=0}^{\ln n}(cn/e^i)^{1/k}\le 2(cn)^{1/k}\sum_{i=0}^\infty e^{-i/k} \le 4k(cn)^{1/k}~.
\]

The strong diameter bound of \lemmaref{lem:diam} holds regardless of which $\beta$ we use, as long as an analogue of \lemmaref{lem:radius} holds. Decreasing the parameter $\beta$ of the exponential distribution increases the probability that a vertex joins a block (so we need less blocks). However, the radius of blocks tend to increase as $\beta$ gets smaller. The following claim implies that the graph is exhausted with high probability.
\begin{claim}
For any vertex $y\in V$, $0\le i\le\ln n$, and $t\in J_i$,
\begin{equation}\label{eq:pr}
\Pr[y\in G_t]\le e^{-2i}~.
\end{equation}
\end{claim}
\begin{proof}
In order to be included in $G_t$, $y$ must not be selected to a block in any phase of any of the stages $0,1,\dots,i-1$.
By \lemmaref{lem:prob}, the probability that $y$ did not join a block in a certain phase of stage $j$ is at most $(1-e^{-\beta_j})$ (even conditioning on anything that happened in previous phases), thus the probability it survived until stage $i$ is at most
\[
\prod_{j=0}^{i-1}(1-e^{-\beta_j})^{s_j}=\prod_{j=0}^{i-1}\left(1-\left(\frac{cn}{e^j}\right)^{-1/k}\right)^{2(cn/e^j)^{1/k}}\le \prod_{j=0}^{i-1}e^{-2}=e^{-2i}~.
\]
\end{proof}

The claim implies (by the union bound), that with probability at least $1-1/n$, there are no remaining vertices after stage $\ln n$. It remains to prove an analogue of \lemmaref{lem:radius}, and argue that with probability at least $1-4/c$, none of the events ${\cal E}_v$ took place. We calculate,
\begin{eqnarray*}
\Pr[{\cal E}_v]&\le&\sum_{i=0}^{\ln n}\sum_{t\in J_i}\Pr[r_v^{(t)}\ge k+1\mid v\in G_t]\cdot\Pr[v\in G_t]\\
&\le&\sum_{i=0}^{\ln n}\sum_{t\in J_i}e^{-\beta_i(k+1)}\cdot e^{-2i}\\
&=&\sum_{i=0}^{\ln n}2\left(\frac{cn}{e^i}\right)^{1/k}\cdot\left(\frac{e^i}{cn}\right)^{1+1/k}\cdot e^{-2i}\\
&=&\frac{1}{n}\sum_{i=0}^{\ln n}\frac{2}{c\cdot e^i}\\
&\le&\frac{4}{cn}~.
\end{eqnarray*}
So by the union bound, with probability at least $1-4/c$, none of events ${\cal E}_v$ occurred, as desired.

\subsection{High Radius Regime}
Note that in \theoremref{thm:main} and \theoremref{thm:less blocks} the number of blocks is $\Omega(\log n)$ for any choice of $k$.
In the regime that $k$, the parameter governing the radius, is larger than $\ln n$, we can get fewer than $\ln n$ blocks.
Concretely, by \claimref{claim:ref} we have that the probability that a vertex $y$ is not in any of the first $\lambda$ blocks is at most $(1-(cn)^{-1/k})^\lambda\le(\ln(cn)/k)^\lambda$ (here we use the estimate $1-e^{-x}\le x$, which is useful when $x$ is small). We would like this probability to be at most $1/cn$, so that the graph will be empty after $\lambda$ phases with probability at least $1-1/c$. To this end, we need
\[
\lambda=\frac{\ln(cn)}{\ln(k/\ln(cn))}~.
\]
In other words, if the number of blocks we want is $\lambda$, then we need to take $k=(cn)^{1/\lambda}\cdot\ln(cn)$, exactly the inverse tradeoff of \theoremref{thm:main}.

\begin{theorem}\label{thm:inverse}
For any unweighted graph $G=(V,E)$ on $n$ vertices, and parameters $1\le \lambda\le \ln n$, $c>3$, our randomized distributed algorithm computes, with probability at least $1-3/c$, a strong $(2(cn)^{1/\lambda}\cdot\ln(cn),\lambda)$ network decomposition of $G$. The number of rounds required is $\lambda(cn)^{1/\lambda}\cdot\ln (cn)$, and each message consists of $O(1)$ words.
\end{theorem}

\section{Acknowledgement}
We are grateful to Nati Linial for discussions that initiated this work.

\bibliographystyle{alpha}
\bibliography{net_decomp}

\end{document}